\documentclass[letterpaper,11pt]{article}
\usepackage[letterpaper,margin=1in]{geometry}
\usepackage[english]{babel}
\usepackage{amsmath}
\usepackage{amssymb}
\usepackage{amsthm}
\usepackage{enumitem}
\usepackage{mathtools}
\usepackage{thm-restate}
\usepackage{algorithm}
\usepackage{algorithmicx}
\usepackage{algpseudocode}
\usepackage{graphicx}
\usepackage[colorlinks=true, allcolors=blue]{hyperref}
\usepackage{hyperref}
\usepackage[nameinlink]{cleveref}
\usepackage{microtype}
\usepackage{xcolor}
\usepackage{xspace}
\usepackage{microtype}
\usepackage{xcolor}
\usepackage{tikz} 
\usepackage{caption}
\usepackage{subcaption}

\usepackage{todonotes}

\newcommand{\local}{\ensuremath{\textsf{LOCAL}}\xspace}

\newcommand{\olocal}{\ensuremath{\textsf{O-LOCAL}}\xspace}
\newcommand{\sleeping}{\ensuremath{\textsf{Sleeping}}\xspace}
\newcommand{\linput}{\ensuremath{\textsf{input}}\xspace}
\newcommand{\loutput}{\ensuremath{\textsf{output}}\xspace}
\newcommand{\bfs}{\ensuremath{\textsf{BFS}}\xspace}

\newtheorem{theorem}{Theorem}

\newtheorem{lemma}[theorem]{Lemma}
\newtheorem{definition}[theorem]{Definition}

\newtheorem{claim}[theorem]{Claim}
\newtheorem{oq}{Open Question}

\hypersetup{
    colorlinks = true,
    allcolors = black
}

\title{Solving Sequential Greedy Problems Distributedly \\ with Sub-Logarithmic Energy Cost}
\author{
Alkida Balliu\thanks{Additional support from MUR (Italy) Department of Excellence 2023 - 2027, and the PNRR MIUR research project GAMING ``Graph Algorithms and MinINg for Green agents'' (PE0000013, CUP D13C24000430001).} \\
\small{Gran Sasso Science Institute}\\
\small{L'Aquila, Italy}
\and
Pierre Fraigniaud\thanks{Additional support from ANR projects DUCAT (ANR-20-CE48-0006) and ENEDISC (ANR-24-CE48-7768-01). }\\
\small{Institut de Recherche en Informatique Fondamentale}\\
\small{CNRS and Université Paris Cité}\\
\small{Paris, France}
\and
Dennis Olivetti\footnotemark[1]\\
\small{Gran Sasso Science Institute}\\
\small{L'Aquila, Italy}
\and
Mikaël Rabie\footnotemark[2]\\
\small{Institut de Recherche en Informatique Fondamentale}\\
\small{CNRS and Université Paris Cité}\\
\small{Paris, France}}

\date{}

\begin{document}
\maketitle

\begin{abstract}
    We study the awake complexity of graph problems that belong to the class \olocal, which includes a subset of problems solvable by sequential greedy algorithms, such as $(\Delta+1)$-coloring and maximal independent set. It is known from previous work that, in $n$-node graphs of maximum degree~$\Delta$, any problem in the class \olocal can be solved by a deterministic distributed algorithm with awake complexity $O(\log\Delta+\log^\star n)$.

    In this paper, we show that any problem belonging to the class \olocal can be solved by a deterministic distributed algorithm with awake complexity $O(\sqrt{\log n}\cdot\log^\star n)$. This leads to a polynomial improvement over the state of the art when $\Delta\gg 2^{\sqrt{\log n}}$, e.g., $\Delta=n^\epsilon$ for some arbitrarily small $\epsilon>0$. The key ingredient for achieving our results is the computation of a network decomposition that uses a small-enough number of colors in sub-logarithmic time in the \sleeping model, which can be of independent interest.

    \medskip

\end{abstract}

\section{Introduction}

In the last four decades, the study of the distributed complexity of graph problems has mainly been focused on understanding the worst-case complexity. Recently, however, different notions of complexities, that better capture the energy usage of a distributed system, have attracted lots of attention. One successful such notion of complexity is captured by the \sleeping variant of the \local model.

The \sleeping model has been introduced in~\cite{chatterjee2020sleeping}. It assumes a set of $n$ fault-free processes connected by an $n$-node graph $G=(V,E)$, in which computation proceeds in lockstep, as a sequence of \emph{synchronous rounds}. At each round, each process, i.e., each node of $G$, is either \emph{awake} or \emph{asleep}. A node $v\in V$ that is awake at a round $r$ can send a message (of arbitrary size) to each of its neighbors, receive the messages sent by its (awake) neighbors, and perform some individual computation (the local computation power of each node is not bounded). Every node $v$ has also the ability to become asleep for a prescribed number of rounds, whose value is chosen by~$v$ as a function of its internal state. If node $v$ remains awake at the end of round~$r\geq 1$, then it directly proceeds to round~$r+1$. Otherwise, it becomes asleep for, say $t\geq 1$ rounds (where $t$ is chosen by~$v$). A sleeping node cannot send messages, and all messages sent to a sleeping node are lost. After $t$ rounds, node~$v$ wakes up, at round $r+t+1$. Initially, i.e., at round~1, all nodes are awake. The main motivation for the study of the \sleeping model is to determine the power and limitation of distributed systems in which processing nodes have the ability to save energy by turning themselves off for a prescribed amount of time. 

In addition to the classical \emph{round complexity} measuring the number of rounds required to solve a given problem, the main measure of complexity in the \sleeping model is the \emph{awake complexity}, which is the maximum, taken over all the nodes $v\in V$, of the number of rounds during which $v$ is awake. 

This paper focuses on the awake complexity of solving a large class of natural problems, namely those in the class \olocal, introduced in \cite{barenboimM21}. Roughly, this class includes all problems which can be solved by a restricted form of sequential greedy algorithms. More specifically, the order in which the nodes are sequentially picked is governed by an arbitrary given acyclic orientation of the edges of the graph~$G$, that is, a node $v$ can be picked only if solutions have been computed for all its descendants according to the given orientation. More importantly, the solution at $v$ can be computed based only on the solutions previously computed for these descendants. Even if \olocal does not contain all sequentially greedily solvable graph problems (e.g., distance-2 coloring is not in \olocal), it includes important problems in distributed network computing, such as $(\Delta+1)$-vertex coloring and maximal independent set (MIS), where $\Delta$ denotes the maximum degree in the graph. 

It is known~\cite{barenboimM21} that every problem $\Pi\in \olocal$ can be solved by a (deterministic) distributed algorithm with awake complexity $O(\log\Delta+\log^\star n)$ in $n$-node graphs with maximum degree~$\Delta$. Having in mind that the function $\log^\star$ grows extremely slowly, this result says that a large class of problems can be solved by having each node awake a number of rounds that is  ``quasi-constant'' in graphs of bounded degree (i.e., for $\Delta=O(1)$). However, for large $\Delta$, say $\Delta=n^\epsilon$ for some $\epsilon>0$, the awake complexity of the distributed algorithm in~\cite{barenboimM21} for \olocal problems grows essentially as $O(\log n)$. 

\subsection{Our Results}

We exhibit a quadratic improvement for the awake complexity of \olocal problems, when expressed as a function of~$n$. Specifically, we show the following result. 

\begin{theorem}\label{thm:olocal-result}
     Any graph problem $\Pi\in\olocal$ can be solved deterministically with awake complexity $O(\sqrt{\log n} \cdot \log^* n)$ in the \sleeping \local model.
\end{theorem}

In particular, for $\Delta\gg 2^{\sqrt{\log n}\cdot \log^\star n}$, e.g., $\Delta=n^\epsilon$ for some arbitrarily small $\epsilon>0$, this improves the algorithm from~\cite{barenboimM21}. The round complexity of our algorithm is $O(n^5 \sqrt{\log n})$. However, if the IDs assigned to the $n$ nodes are taken from $\{1,\ldots,n\}$, then this round complexity improves to $O(n^2 \sqrt{\log n})$. More generally,  if the IDs are taken from $\{1,\ldots,n^c\}$ with $c\geq 1$, then the round complexity of our algorithm becomes $O(n^{\min\{1+c,5\}} \sqrt{\log n})$.

Our result is obtained thanks to an interplay between two similar forms of \emph{network decomposition}. Both are partitioning the nodes of the input graph $G=(V,E)$ into \emph{clusters} (i.e., connected induced subgraphs of~$G$). Every cluster $C$ has a  root~$r$, and a label~$x$. It is encoded distributedly by having each node $v\in C$ storing a pair $(\ell(v),\delta(v))$ where $\ell(v)=x$, and $\delta(v)$ is the distance from $v$ to $r$ in~$C$. We thus refer to our network decomposition as \emph{\bfs-clustering}. Our two forms of clustering differ according to the labeling of the clusters. The \emph{uniquely-labeled} \bfs-clustering requires that no two clusters have the same label. Instead, the \emph{colored} \bfs-clustering allows different clusters to have the same label, as long as they are not adjacent. That is, clusters $C$ and $C'$ of $G=(V,E)$ may have the same label, but then it must be the case that, for every $v\in C$ and every $v'\in C'$, $\{v,v'\}\notin E$.

The former type of clustering enables induction. In fact, in the virtual graph resulting from merging each cluster into a single vertex, it is possible to run algorithms for the \local model with only a constant overhead on the awake complexity. The latter type of clustering is weaker, yet it is sufficient to apply previous results of the literature on that form of clustering, such as the one in~\cite{barenboimM21}, which assumes a given proper $k$-coloring.  

We establish Theorem~\ref{thm:olocal-result} in two steps. First, we show in Theorem~\ref{thm:solving-given-nd} that, for any $n$-node graph $G=(V,E)$, and any $\Pi\in \olocal$, if a colored  \bfs-clustering of~$G$ is given to the nodes, then $\Pi$ can be solved by a distributed algorithm with awake complexity $O(\log c)$, and round complexity $O(c \cdot n)$ where $c$ is the range of colors used by the colored \bfs-clustering. Second, we show in Theorem~\ref{thm:nd} that, for any $n$-node graph $G= (V,E)$, a  colored  \bfs-clustering with $2^{O(\sqrt{\log n})}$ colors can be computed by a distributed algorithm with awake complexity $O(\sqrt{\log n} \cdot \log^* n)$, and round complexity $O(n^5 \sqrt{\log n})$. Theorem~\ref{thm:olocal-result}  follows directly from these two theorems. 

\subsection{Additional Related Work}

The \sleeping model was introduced in~\cite{chatterjee2020sleeping}, which opened a large avenue of research dedicated to measuring the potential benefit of providing the nodes of a distributed system with the ability to turn on and off at will. 
In particular, MIS and $(\Delta+1)$-coloring have attracted a lot of attention in this context (see \cite{chatterjee2020sleeping,GhaffariP2022,GhaffariP2023,HouraniPR22,DufoulonFRZ2024,barenboimM21}).
The aforementioned paper~\cite{barenboimM21} on solving \olocal problems established that both $(\Delta+1)$-coloring and MIS have deterministic awake complexity  $O(\log\Delta+\log^\star n)$.
A direct consequence of the result in our paper is that MIS and $(\Delta+1)$-coloring have deterministic awake complexity $O(\sqrt{\log n}\cdot\log^\star n)$.
Regarding randomized algorithms, it was shown that MIS and  $(\Delta+1)$-coloring have awake complexity $O(\log\log n)$ \cite{DufoulonMP23,DufoulonFRZ2024}.
This presents an improvement over the best known randomized algorithms for $(\Delta+1)$-coloring in the \local model~\cite{ChangLP18,HalldorssonKNT22}, which have round complexity $\tilde{O}(\log^2\log n)$, though to the expense of a significant increase of the round complexity. 
Several other problems have been studied in the \sleeping model, including matching~\cite{GhaffariP2022}, vertex cover~\cite{GhaffariP2022}, spanning tree~\cite{barenboimM21}, and minimum-weight spanning tree~\cite{AugustineMP24}. 

We note that notions that are closely related to the one of awake complexity have been investigated in the context of radio networks (see e.g., \cite{jurdzinski2002energy, chang2019exponential, nakanoO00a, kardasKP13, jurdzinski2002efficient, chang2020energy, benderKPY18}), and also in the context of the \local model, such as the line of work that studies the average complexity of graph problems in the \local model (see e.g., \cite{barenboim2019distributed, feuilloley2020avg, balliu0KOS24, balliu0KOS23, balliu0KO22}).

\section{Model and Definitions}

\subsection{The \sleeping Model}

We consider the \sleeping variant of the standard \local model~\cite{Peleg2000}. Recall that the \local model assumes $n$ nodes connected as an $n$-node network, modeled as a simple connected graph $G=(V,E)$. Every node has an identifier, which is unique in the network. The identifiers are supposed to be taken from a polynomial range of identifiers, and thus every identifier can be stored on $O(\log n)$ bits in $n$-node graphs. Computation proceeds synchronously as a sequence of rounds. All nodes starts at round~1. At each round, every node can send a message to each of its neighbors, receive the messages from its neighbors, and perform some internal computation. There are no limits on the size of the messages, nor on the amount of internal computation a node may perform at each round. 

At each round of a distributed algorithm for the \local model, the \sleeping model provides each node $v$ with the additional capacity to become ``asleep'' (i.e., to turn off) for a prescribed number of rounds, which is chosen at will by~$v$. A node  executing an algorithm in the \sleeping model thus alternates between sequences of rounds during which it is ``awake'' (i.e., is on), and sequences of rounds during which it is asleep. If a node is awake at a given round, it acts like in the \local model. However, a node $v$ that is asleep at a round is totally inactive at this round. In particular, all messages that may have been sent to $v$ by its neighbors at this round are lost, and they cannot be recovered (unless sent again by the neighbors at some later round during which $v$ is awake). Initially, all nodes are awake. We assume that every processing node $v$ initially knows the order $n$ of the graph $G=(V,E)$ it belongs to.

The \emph{round complexity} of an algorithm $A$ running under the \sleeping model is the maximum, taken over all graphs $G$ of at most $n$ nodes, and over all nodes $v$ of $G$, of the  number of rounds required for $v$ to terminate. The \emph{awake complexity} of $A$, is the maximum, taken over all graphs $G$ of at most $n$ nodes, and over all nodes $v$ of $G$, of the  number of rounds during which $v$ is awake. 

\subsection{The \olocal Class of Graph Problems}

\olocal,  which stands for Oriented-\local~\cite{barenboimM21}, is a class of graph problems whose every solution consists of an appropriate labeling of the nodes of the input graph (e.g., $(\Delta+1)$-coloring, maximal independent set (MIS), etc.) that can be computed by a sequential greedy algorithm performing as follows. Let $G = (V,E)$ be a (non-directed) graph, let $\mu$ be an arbitrary acyclic orientation of the edges of~$G$, and let $G_\mu$ be the resulting directed graph. For every $v \in V$, let $G_{\mu}(v)$ be the subgraph of $G_\mu$ induced by $v$ and all the nodes that can be reached from $v$ by following outgoing edges only. That is, $G_{\mu}(v)$ contains all nodes $w$ such that there is a path from $v$ to $w$ in the directed graph~$G_\mu$. The greedy algorithm picks nodes in arbitrary order, but respecting the orientation~$\mu$. That is, a node $v$ can be picked only if the output of all the nodes in $G_\mu(v)\smallsetminus \{v\}$ has been previously computed by the algorithm. Let $v$ be such a node. The algorithm must be able to compute a correct output for $v$ based only on the outputs computed previously by all nodes in $G_\mu(v) \setminus \{v\}$. A graph problem $\Pi$ is in \olocal if there is a greedy algorithm that succeeds to compute a solution for $\Pi$ for every input graph~$G$, and for every acyclic orientation~$\mu$ of the edges of~$G$. 

The class \olocal is, by definition, a subset of all graph problems that can be solved by a greedy algorithm fixing the outputs of the nodes by treating them sequentially, in arbitrary order. In particular, \olocal contains problems at the core of distributed computing in networks, including (proper) $(\Delta+1)$-coloring of graphs with maximum degree~$\Delta$, and maximal independent set (MIS). Distance-2 $(\Delta^2+1)$-coloring is however not in \olocal, as witnessed by the $n$-node path~$P$, with $n\geq 6$. Indeed, let $\mu$ be the edge orientation of $P$ such that every two incident edges are directed in opposite direction. For this orientation, the color of every node $v$ with out-degree~0 in $P_\mu$ must be fixed without any other information than the ID of~$v$. For every $f:\{1,\dots,n\}\to \{1,\dots,5\}$, there is an assignment of IDs to the nodes of $P$ such that the color assignment $f$ collides for two vertices at distance~2 in~$P$.  

\subsection{Definitions}

We define a collection of concepts that will be at the core of the technical parts of the paper. Specifically, we define two forms of network decomposition, a.k.a.~clustering, and we associate a virtual graph to each of these two decompositions, referred to as \emph{uniquely-labeled \bfs-clustering} and  \emph{colored BFS-clustering}. \bfs refers to the fact that, in each cluster, there is a special node, called the root, and every node in the cluster knows its distance to the root. A uniquely-labeled \bfs-clustering assigns a label to each cluster, and this label is unique, i.e., no two clusters have the same label. This enables to recurse in the virtual graph resulting from collapsing every cluster into a single vertex, and assigning the label of the cluster as ID of the vertex. If one enforces conditions on the number of labels, computing a uniquely-labeled \bfs-clustering is not easy. In such a case, colored \bfs-clusterings are easier to compute, but two distinct cluster may be given the same label (now refer to as color), as long as there are no edges between them in the graph. The issue with colored BFS-clusterings is that one may not be able to recurse in the virtual graphs induced by these \bfs-clustering, whenever considering algorithms in the \local model, which assume distinct IDs, yet they are sufficient for applying known results (e.g., algorithms in~\cite{barenboimM21}, which assume a given $k$-coloring). We thus play with the two types of \bfs-clustering for benefiting of the best of each of them. 

\begin{definition}\label{def:bfs-cluster}
     A \emph{uniquely-labeled  \bfs-clustering} of a graph $G=(V,E)$ is a pair of functions $(\ell,\delta)$ assigning a pair $(\ell(v),\delta(v))\in \mathbb{N}\times\mathbb{N}$ to each node $v\in V$ such that, for every integer $i>0$, if the subgraph $G_i$ of $G$ induced by the nodes in $\{v\in V\mid \ell(v)=i\}$ is non-empty, then $G_i$ is connected, there is a unique node $u$ of $G_i$ with $\delta(u)=0$, and, for every node $v$ of $G_i$, $\delta(v)$~is the distance  from $u$ to $v$ in~$G_i$. 
\end{definition}

Given a uniquely-labeled  \bfs-clustering $(\ell,\delta)$ of $G=(V,E)$, and $i\in \mathbb{N}$, the set of nodes $\{v\in V\mid \ell(v)=i\}$ is called a \emph{cluster}, $\ell(v)$ is called the \emph{cluster-ID} of~$v\in V$, and we somewhat abuse terminology by saying that node $v$ with $\ell(v)=i$ belongs to the cluster~$i$, or the $i$-th cluster.

\begin{definition}
    The \emph{virtual graph induced by a uniquely-labeled  \bfs-clustering} $(\ell,\delta)$ of a graph $G=(V,E)$ is the graph $H$ with vertex-set equal to the set $\{\ell(v)\mid v\in V\}$ of cluster-IDs, and two vertices $i$ and $j$ of~$H$ are neighbors in~$H$ if there exists an edge $\{u,v\}\in E$ such that $\ell(u)=i$ and $\ell(v)=j$. 
\end{definition}

That is, $H$ can be viewed as the graph obtained from $G$ by merging each cluster  of the  uniquely-labeled  \bfs-clustering into one vertex, and replacing every set of parallel edges resulting from  this merge by a single edge. 

\begin{definition}
    A \emph{ colored  \bfs-clustering} of a graph $G=(V,E)$ is a pair of functions $(\gamma,\delta)$ assigning a pair $(\gamma(v),\delta(v))\in \mathbb{N}\times\mathbb{N}$ to every node $v \in V$ such that, for every $i\in \mathbb{N}$, the subgraph $G_i$ of $G$ induced by the nodes $\{v\in V\mid \gamma(v)=i\}$ satisfies that, for every connected component $C$ of~$G_i$, there is a unique node $u$ of $C$ with $\delta(u)=0$, and, for every $v\in C$, $\delta(v)$ is the distance between $u$ and $v$ in $C$.
\end{definition}

Given a  colored  \bfs-clustering $(\gamma,\delta)$ of a graph $G=(V,E)$, for every $i\in \mathbb{N}$, and every connected component of the subgraph $G_i$ of~$G$ induced by the nodes $\{v\in V\mid \gamma(v)=i\}$, we say that $C$ is a \emph{cluster}, and that $i$ is the \emph{color} of that cluster. Note that, in a  uniquely-labeled  \bfs-clustering of a graph~$G$, each set of nodes with identical label is connected, and the ID of the corresponding cluster is unique (no other clusters have the same ID). In contrast, in a  colored  \bfs-clustering, non-adjacent clusters are allowed to have the same~ID, which motivates the terminology ``color''.

\begin{definition} 
The \emph{virtual graph induced by a  colored  \bfs-clustering} of a graph $G=(V,E)$ is the graph $H$ with vertex-set equal to the set of clusters resulting from the clustering,  and two vertices $C$ and $C'$ of $H$ are neighbors in $H$ if there exist $\{u,v\}\in E$ such that  $u\in C$ and $v\in C'$.
\end{definition}

As for  uniquely-labeled  \bfs-clustering, the virtual graph $H$ induced by a  colored  \bfs-clustering can be viewed as the graph obtained from $G$ by merging each cluster $C$ of the clustering into one vertex, and replacing every set of parallel edges resulting from  this merge by a single edge. 

\section{Prelimilaries}

In this section, we prove some useful statements about the Sleeping model.
We start by stating the following lemma, which is implied by Lemma~2.1 in~\cite{barenboimM21}, but for which we provide a sketch of proof that can serve as a warm up for the reader unfamiliar with the Sleeping model.

\begin{lemma}[Barenboim and Maimon \cite{barenboimM21}]\label{lem:broad-converge-cast}
    Let $G=(V, E)$ be a graph, and let $T$ be a spanning tree of $G$ rooted at $r\in V$. For every $v \in V\smallsetminus \{r\}$, let $p(v)$ denote the parent of $v$ in~$T$. Let $L:V\to \{1,\dots,N\}$ be a labeling such that $L(v) > L(p(v))$ for every $v \in V\smallsetminus \{r\}$. Let us assume that every node $v$ is given its parent $p(v)$ (if $v\neq r$), its label $L(v)$, and~$N$. Each of the two tasks broadcast and  convergecast can be performed in $G$ under the Sleeping model, with awake complexity~3, and round complexity~$O(N)$.
\end{lemma}

\begin{proof}[Sketch of proof]
Let $M$ be the message to be broadcasted by $r$ to all the nodes. At the first round, every node $v$ communicates with its neighbors for learning the label $L(p(v))$ of its parent, and becomes idle. Every node then wakes up only twice. Specifically, every node $v$ wakes up at round $2 + L(p(v))$ during which it communicates with its neighbors for learning~$M$ (that it receives from its parent). Node $v$ then becomes idle, and wakes up again at round $2 + L(v)$ during which it sends $M$ to all its neighbors, and thus in particular to its children in~$T$.

Convergecast is performed in a way similar to broadcast. For every node $v\in V$, let $M_v$ be the message of~$v$. First, every node $v$ assigns itself a new label $L'(v)= N - L(v)$. Note that this new labeling  satisfies $L'(v)<L'(p(v))$ for every $v\neq r$. At the first round, every node~$v$ communicate with its neighbors for learning the label $L'(p(v))$ of its parent, and becomes idle. As for broadcast, every node then wakes up only twice.  Specifically, node $v$ wakes up at round $2 + L'(v)$ during which it receives from the children all convergecast messages from nodes in the subtrees of $T$ rooted at these children.  Node $v$ then becomes idle, and wakes up again at round $2 + L'(p(v))$ during which $v$ forwards to its parent all the messages it collected during the previous rounds.  
\end{proof}

The following lemma establishes a tight connection between a  uniquely-labeled  \bfs-clustering of a graph $G$ and the virtual graph $H$ resulting from that clustering. 

\begin{lemma}\label{lem:simulation-in-virtual-graph}
Let $G=(V,E)$ be an $n$-node graph, and let $H$ be  the virtual graph induced by some  uniquely-labeled  \bfs-clustering $(\ell,\delta)$ of $G$.
Let $\mathcal{A}$ be a distributed algorithm running on $H$ with awake complexity~$\alpha$ and round complexity~$\varrho$. For every vertex $i$ of~$H$, let $\linput(i)$ be the input of vertex $i$ of $H$, and let $\loutput(i)$ be the output of $\mathcal{A}$ at vertex~$i$. Assuming each node $v$ of $G$ knows $\linput(\ell(v))$, it is possible to run $\mathcal{A}$ on $G$ such that every node $v$ of $G$ computes $\loutput(\ell(v))$ with awake complexity at most~$7\cdot \alpha$ and round complexity $O(\varrho \cdot n)$.
\end{lemma}

\begin{proof}
    We prove that each round that $\mathcal{A}$ performs on $H$ can be simulated on $G$ by having each node of $G$ awake for at most $7$ rounds during an interval of $R = O(n)$ rounds, where $R$ is the round complexity of computing a convergecast and a broadcast (see \Cref{lem:broad-converge-cast}).
    Let $v \in V$ and let $i = \ell(v)$. Node $v$ of $G$ simulates node $i$ of $H$ as follows. 
    Let us call a phase a series of $R$ consecutive rounds. We prove by induction on $k$ that, (1)~after $k$ phases (i.e., after $k\cdot R$ rounds), node $v$ is able to compute the state that node $i$ of $H$ would have by executing $k$ rounds of $\mathcal{A}$, and (2)~if up to this point node $i$ has been awake for $r$ rounds, then $v$ has been awake for at most $7r$ rounds.

    The base case $k=0$ trivially holds. Suppose that the claim holds for all $k'$ with $0\leq k' < k$, and let us prove the claim for~$k$. We distinguish two cases, depending on whether $i$ was awake at round $k$ or not. In the latter case, $v$ does not wake up at all during phase~$k$, and the claim trivially holds. In the former case, $v$ wakes up during the first round of phase $k$, and shares its state with all its (awake) neighbors in~$G$. If $\delta(v) > 0$, then $v$ picks  as parent any neighbor $u$ that belongs to the same cluster as~$v$, and satisfying $\delta(v) = \delta(u) + 1$. Then, $v$ performs a convergecast, followed by a broadcast, as specified in \Cref{lem:broad-converge-cast}. Hence, in total, $v$~was awake during at most $7$~rounds during phase~$i$. During that phase, $v$~learned the state of all the neighbors of~$i$ in~$H$, and thus it can compute the new state that node $i$ has computed at round~$k$.
\end{proof}

We now observe that, similarly to what holds in \local model, the execution of two algorithms can be concatenated into one, assuming that all nodes know when the first algorithm terminates. 

\begin{lemma}\label{lem:concatenate}
    Let $\mathcal{A}_1$ and $\mathcal{A}_2$ be two algorithms with respective awake complexities $S_1$ and~$S_2$, and respective round complexities $T_1$ and $T_2$.
    The algorithm obtained by running the two algorithms consecutively, first $\mathcal{A}_1$ for $T_1$ rounds, and then $\mathcal{A}_2$ for $T_2$ rounds, has awake complexity $S_1 + S_2$ (and, by design, round complexity $T_1 + T_2$).
\end{lemma}

\begin{proof}
Nodes first execute $\mathcal{A}_1$. When $\mathcal{A}_1$ terminates at node~$v$, sait at round $t\leq T_1$, it becomes idle for $T_1-t$ rounds, until round $T_1 + 1$. At round $T_1 + 1$ all nodes start executing $\mathcal{A}_2$. For every node~$v$, the number of rounds during which $v$ is awake is thus $S_1+S_2$.
\end{proof}

We will use \Cref{lem:concatenate} implicitly throughout the paper, without systematically referring to it.

\section{Solving any \olocal Problem Given a Colored \bfs-Clustering}

In this section, we prove that, if the nodes of a graph $G=(V,E)$ are given a colored \bfs-clustering of $G$ as input (i.e., every node knows its color, and its distance value in the clustering), then every problem in the \olocal class can be solved efficiently in~$G$. More formally, this section is entirely dedicated to proving the following statement.

\begin{theorem}\label{thm:solving-given-nd}
    Let $G=(V,E)$ be an $n$-node graph and let $\Pi$ be a problem in \olocal. Let $(\gamma,\delta)$ be a  colored  \bfs-clustering of~$G$, and let us assume that every node $v\in V$ initially knows its color $\gamma(v)$, and its distance~$\delta(v)$. Let $c=\max_{v\in V}\gamma(v)$. $\Pi$ can be solved by a distributed algorithm with awake complexity $O(\log c)$, and round complexity $O(c \cdot n)$.
\end{theorem}

Note that the complexity of the algorithm in the theorem above depends on $\max_{v\in V}\gamma(v)$ but not on~$\max_{v\in V}\delta(v)$. The proof of the theorem is based on a result of \cite{barenboimM21} stating that every problem in \olocal has awake complexity $O(\log \Delta + \log^* n)$ in the class of graphs with at most $n$ nodes, and maximum degree at most~$\Delta$. Let us first briefly provide the reader with a high-level intuition of this result, as it might be useful to the reader for a better understanding of the main ideas allowing us to establish Theorem~\ref{thm:solving-given-nd}. The algorithm from \cite{barenboimM21} works as follows. First, the nodes compute an $O(\Delta^2)$-coloring of the graph, in $O(\log^* n)$ rounds, by using Linial's coloring algorithm~\cite{linial92}. Let $q = O(\Delta^2)$ be the smallest power of $2$ larger than the number of colors produced by Linial's algorithm. Second, the colors in $\{1,\ldots, q\}$ are mapped to the larger color palette $\{1,\ldots,2q-1\}$. This is done by using a carefully crafted mapping $\phi:\{1,\ldots, q\}\to \{1,\ldots,2q-1\}$ satisfying the following properties.

\begin{lemma}[Barenboim and Maimon \cite{barenboimM21}]\label{lem:mapping-and-function}
    There exist two mappings 
    \[
    \phi : \{1,\ldots,q\} \rightarrow \{1,\ldots,2q-1\}, \text{ and }
    r : \{1,\ldots,q\} \rightarrow 2^{\{1,\ldots,2q-1\}}
    \]
    satisfying the following: (1)~for every $c \in \{1,\ldots,q\}$, $| r(c) | = 1 + \log q$, (2)~$\phi(c) \in r(c)$, and (3)~for every pair of distinct colors $(c_1,c_2) \in \{1,\ldots,q\}^2$, there exists $x \in r(c_1) \cap r(c_2)$ such that $$\min\{\phi(c_1),\phi(c_2)\} < x < \max\{\phi(c_1),\phi(c_2)\}.$$
\end{lemma}

\Cref{lem:mapping-and-function} is proved by considering a binary tree $T$ whose nodes are the elements of the set $\{1,\ldots,2q-1\}$ labeled according to a in-order traversal  (see Figure~\ref{fig:tree}). The leaves of $T$ are the images by~$\phi$ of the elements in  $\{1,\ldots,q\}$, i.e., $\phi(c)$ is the label of the leaf with the $c$-th smallest label among all leaves. For every color~$c\in \{1,\ldots,q\}$, $r(c)$ is the set of labels of the nodes in the path from the root of $T$ to the leaf~$\phi(c)$.

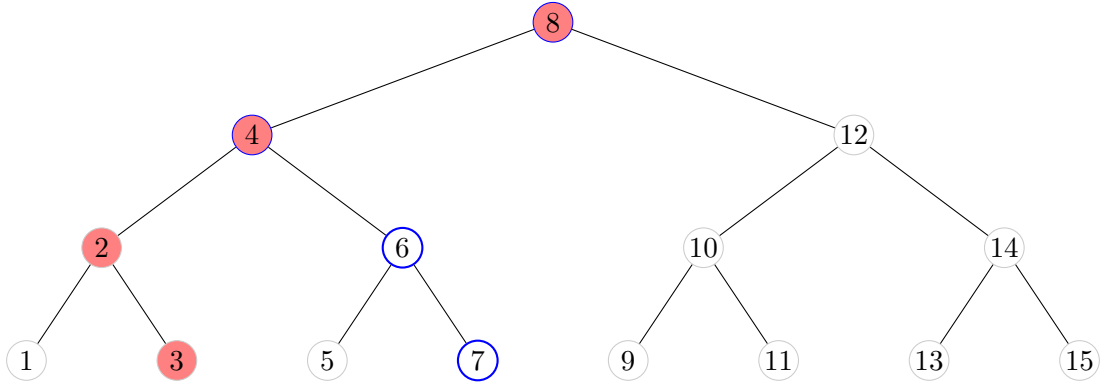
\begin{figure}
\begin{center}
\begin{tikzpicture}[square/.style={regular polygon,regular polygon sides=4}]
   \tikzstyle{circlenode}=[draw,circle,minimum size=70pt,inner sep=0pt]
    \tikzstyle{whitenode}=[draw=black!20,circle=black!20,fill=white,minimum size=15pt,inner sep=0pt]
    \tikzstyle{bignode}=[draw,circle,fill=white,minimum size=15pt,inner sep=0pt]
    \tikzstyle{squarenode}=[draw,square,fill=white,minimum size=15pt,inner sep=0pt]
    \tikzstyle{rednode}=[draw=black!20,circle=black!20,fill=red!50,minimum size=15pt,inner sep=0pt]
    \tikzstyle{bluenode}=[draw=blue,circle=blue, thick,fill=white,minimum size=15pt,inner sep=0pt]
    \tikzstyle{bluerednode}=[draw=blue,circle=blue,fill=red!50,minimum size=15pt,inner sep=0pt]
    \tikzstyle{fullrednode}=[draw=red,circle=red,fill=red,minimum size=15pt,inner sep=0pt]

\draw (0,0) node[bluerednode] (a8) []{8};
\draw (-4,-1.5) node[bluerednode] (a4) []{4};
\draw (4,-1.5) node[whitenode] (a12) []{12};
\draw (-6,-3) node[rednode] (a2) []{2};
\draw (-2,-3) node[bluenode] (a6) []{6};
\draw (2,-3) node[whitenode] (a10) []{10};
\draw (6,-3) node[whitenode] (a14) []{14};
\draw (-7,-4.5) node[whitenode] (a1) []{1};
\draw (-5,-4.5) node[rednode] (a3) []{3};
\draw (-3,-4.5) node[whitenode] (a5) []{5};
\draw (-1,-4.5) node[bluenode] (a7) []{7};
\draw (1,-4.5) node[whitenode] (a9) []{9};
\draw (3,-4.5) node[whitenode] (a11) []{11};
\draw (5,-4.5) node[whitenode] (a13) []{13};
\draw (7,-4.5) node[whitenode] (a15) []{15};

	\draw (a8) edge [] node {} (a4);
	\draw (a8) edge [] node {} (a12);
	\draw (a2) edge [] node {} (a4);
	\draw (a6) edge [] node {} (a4);
	\draw (a10) edge [] node {} (a12);
	\draw (a14) edge [] node {} (a12);
	\draw (a2) edge [] node {} (a1);
	\draw (a6) edge [] node {} (a5);
	\draw (a10) edge [] node {} (a9);
	\draw (a14) edge [] node {} (a13);
	\draw (a2) edge [] node {} (a3);
	\draw (a6) edge [] node {} (a7);
	\draw (a10) edge [] node {} (a11);
	\draw (a14) edge [] node {} (a15);

\end{tikzpicture}
\caption{The tree used in the proof of Lemma~\ref{lem:mapping-and-function}. We have $\phi(2)=3$ as the second smallest label of a leaf is~3, and $r(2)=\{2,3,4,8\}$. Similarly, $\phi(4)=7$, and $r(4)=\{4,6,7,8\}$.
Note that the lowest common ancestor of the nodes labeled 3 and 7 is the node labeled~4, and indeed $3<4<7$.}
\label{fig:tree}
\end{center}
\end{figure}

Given the computed $q$-coloring, and given the mapping $\phi$ and $r$, the proof by \cite{barenboimM21} that every problem in \olocal has awake complexity $O(\log \Delta + \log^* n)$ continues as follows. 
Let $G=(V,E)$ be a graph, and let $\mu$ be the orientation of the edges of $G$ resulting from the proper $q$-coloring of~$G$, that is, every edge is oriented from its end-point of higher color to its other end-point.
For every color $c \in\{1,\dots,q\}$, let us consider the two sets
\[
r_{<}(c) = \{ x \in r(c) \mid x < \phi(c) \},
\]
and
\[
r_{>}(c) = \{ x \in r(c) \mid  x > \phi(c) \}.
\]
Every node $v\in V$ with color $c=c(v)$ wakes up at every round $i \in r(c)$, and performs the following operations (recall that $\phi(c)\in r(c)$):
\begin{itemize}
    \item If $i \in  r_{<}(c)$, then node $v$ receives the states of all its neighbors awaken at round~$i$, and stores this information in its own state;
    \item If $i = \phi(c)$, then node $v$ decides its own solution as a function of its  state (which possibly includes the states of some neighbors received during previous rounds);
    \item If $i \in r_{>}(c)$, then node $v$ sends its state to its neighbors.
\end{itemize}
The awake complexity of this algorithm is $O(\log q)$ as, from \Cref{lem:mapping-and-function}, $| r(c) | = 1 + \log(q)$. 

To establish that the algorithm correctly produces a solution for the problem $\Pi$, it is sufficient to show that each node $v$ knows the whole graph~$G_\mu(v)$ when it must compute its local solution to~$\Pi$. This can be shown by a simple induction on the colors in $\{1,\dots,q\}$. Indeed, this clearly holds for nodes of color~$1$.
Now, let us assume that each neighbor $u$ of $v$ with color $c(u)<c(v)$ knows $G_\mu(u)$ at round $\phi(c(u))$.   
By \Cref{lem:mapping-and-function}, there exists an element $x \in r(c(u)) \cap r(c(v))$ such that $\phi(c(u)) < x < \phi(c(v))$. It follows that, at round~$x$, node $v$ learns $G_\mu(u)$, and thus,  node $v$ knows $G_\mu(v)$ at round $\phi(c(v))$.
We summarize the result of \cite{barenboimM21} by the following lemma for further references in the text.

\begin{lemma}[Barenboim and Maimon \cite{barenboimM21}]
\label{lem:solve-olocal-given-coloring}
    For every problem $\Pi \in \olocal$, there exists a distributed algorithm that solves $\Pi$ in every graph $G=(V,E)$ provided with a proper $k$-coloring of its nodes,   running in $O(k)$ rounds, and with awake complexity $O(\log k)$.
\end{lemma}

\Cref{thm:solving-given-nd} essentially states that \Cref{lem:solve-olocal-given-coloring} also holds whenever, instead of a proper coloring of the graph, the algorithm is given a  colored  \bfs-clustering of the graph. 

\begin{proof}[Proof of \Cref{thm:solving-given-nd}]
    Let $G = (V,E)$ be an $n$-node graph and let $\Pi\in\olocal$. Let $(\gamma,\delta)$ be a  colored  \bfs-clustering of~$G$, and let us assume that every node $v\in V$ initially knows its color $\gamma(v)$, and its distance~$\delta(v)$. Let $H$ be the virtual graph induced by the  colored  \bfs-clustering $(\gamma,\delta)$. For every vertex $i$ of $H$, the nodes of each cluster~$C$ labeled~$i$ use the broadcast algorithm of \Cref{lem:broad-converge-cast} to learn the ID of the root of~$G_i$. (For doing so, every node $v\in C$ sets its parent $p(v)$ as one of its neighbors $u\in C$ with $\delta(u)<\delta(v)$.) After this phase, $H$ can then indifferently be viewed as either the virtual graph induced by the given  colored  \bfs-clustering $(\gamma,\delta)$ of $G$, or as the virtual graph induced by the  uniquely-labeled  \bfs-clustering $(\ell,\delta)$ of $G$ where $\ell$ is the labeling of the clusters of $(\gamma,\delta)$ defined by the IDs of the roots of these clusters. 

    Let us now define a problem $\Pi' \in \olocal$ with the property  that any solution for $\Pi'$ in the virtual graph $H$ can be converted into a solution for $\Pi$ in~$G$.
    The problem $\Pi'$ on a graph $H$ is defined as follows. The input for $\Pi'$ is defined w.r.t.\ a graph $G$ satisfying that $H$ is the virtual graph induced by some  uniquely-labeled  \bfs-clustering $(\ell,\delta)$ of~$G$.
    For each vertex $i$ of~$H$, the input of $i$ is the subgraph $G_i$ of~$G$ (including the node IDs, and the node inputs, if any) induced by the nodes of $G$ belonging to the cluster~$i$, plus the edges incident to at least one node of $G_i$ and going out of~$G_i$. Solving $\Pi'$ on $H$ requires each vertex $i$ of $H$ to output a set of individual solutions for $\Pi$ in $G$ potentially produced at all the nodes of~$G_i$ such that the union of all these sets, taken over all vertices of~$H$, is a valid solution for~$\Pi$ in~$G$. By construction, a solution for $\Pi'$ in $H$ directly provides a solution for $\Pi$ in $G$.
       
    \begin{claim}\label{claim:Pi-prime-in-OLOCAL}
        $\Pi' \in \olocal$.
    \end{claim} 
    
    For establishing the claim, by the definition of \olocal, it is sufficient to prove that, assuming a given an acyclic orientation $\mu_H$ of the edges of $H$, then, for every vertex $i$ of $H$, if a solution for $\Pi'$ has already been computed at every vertex of $H_{\mu_H}(i) \smallsetminus \{i\}$, then one can compute the solution for vertex $i$ of $H$. For this purpose, we define an acyclic orientation $\mu_G$ of the edges of $G$ as follows. Any edge $\{u,v\}$ connecting a node $u$ of a cluster $G_i$ of the  uniquely-labeled  \bfs-clustering $(\ell,\delta)$ and a node $v$ of a cluster $G_j$ is oriented according to the orientation in $\mu_H$ of the edge $e=\{i,j\}$ of~$H$, i.e., from $u$ to $v$ if $e$ is oriented from $i$ to $j$ in $\mu_H$, or from $v$ to $u$ otherwise. Any edge $\{u,v\}$ connecting two nodes in a same cluster $G_i$ of $(\ell,\delta)$ is oriented in $\mu_G$ using $\delta$, i.e., from $u$ to $v$ if $\delta(u)<\delta(v)$, and from $v$ to $u$ otherwise. Let $i$ be a vertex of $H$, and let us assume a solution for $\Pi'$ has already been computed at every vertex of $H_{\mu_H}(i) \smallsetminus \{i\}$. Since every inter-cluster edge of $G$ is oriented in $\mu_G$ the same as the associated edge of $H$ is oriented in~$\mu_H$, this assumption implies in particular that a solution for $\Pi$ has been computed for all nodes of $G$ in the clusters $G_j$ for all vertices $j$ in $H_{\mu_H}(i) \smallsetminus \{i\}$. To compute a solution for~$i$, we go over all the nodes in $G_i$ in the order imposed by $\mu_G$, i.e., in the order imposed by~$\delta$.  This enables to compute a solution of $\Pi$ for all nodes in~$G_i$. By definition of $\Pi'$, this results in a solution of $\Pi'$ for vertex~$i$. Therefore $\Pi' \in \olocal$, which completes the proof of Claim~\ref{claim:Pi-prime-in-OLOCAL}.

\bigbreak 

    Thanks to Claim~\ref{claim:Pi-prime-in-OLOCAL}, we can apply \Cref{lem:solve-olocal-given-coloring} on~$\Pi'$. We do so by using the proper coloring~$\gamma$ of~$H$, and \Cref{lem:solve-olocal-given-coloring} says that $\Pi'$ can be solved in $H$ by a distributed algorithm $A$ with awake complexity $O(\log c)$, and round complexity~$O(c)$. We use \Cref{lem:simulation-in-virtual-graph} for simulating $A$ on~$G$ using the uniquely-labeled \bfs-clustering $(\ell,\delta)$.
     Each round of $A$ in $H$ is a virtual round in~$G$. Each virtual round is replaced by a sequence of $O(n)$ rounds used for (1)~gathering at the root of each cluster $G_i$ of $(\ell,\delta)$ the set of all messages exchanged between $G_i$ and its neighboring clusters during the previous virtual round, and (2)~broadcasting this set to all nodes in~$G_i$. The total number of rounds consumed for the simulation of $A$ in $G$ is therefore $O(nc)$. Each vertex $i$ of $H$ is awake $O(\log c)$ times during the execution of~$A$, and, whenever $i$ is awake, each node of $G_i$ is awake a constant number of rounds for performing the convergecast-broadcast operation. It follows that the awake complexity of the simulation of $A$ in $G$ remains constant. This completes the proof of \Cref{thm:solving-given-nd}.
\end{proof}

\section{Computing a Colored \bfs-Clustering}

In the previous section, we have seen how to solve any problem in \olocal assuming given a  colored  \bfs-clustering. In this section, we show how to compute such a colored \bfs-clustering efficiently in the \sleeping model. Specifically, this section is entirely dedicated to establishing the following theorem. Note that \Cref{thm:olocal-result}, which is the main contribution of the paper, directly follows from combining \Cref{thm:nd} with \Cref{thm:solving-given-nd}. 

\begin{theorem}\label{thm:nd}
    Let $G= (V,E)$ be an $n$-node graph. A  colored  \bfs-clustering $(\gamma,\delta)$ of $G$ with $\max_{v\in V} \gamma(v)= 2^{O(\sqrt{\log n})}$ can be computed by a distributed algorithm with awake complexity $O(\sqrt{\log n} \cdot \log^* n)$, and round complexity $O(n^5 \sqrt{\log n})$. 
\end{theorem}

\paragraph{Remark.}

As it will appear clear in the proof, if the IDs assigned to the $n$ nodes are taken from $\{1,\ldots,n\}$, then the round complexity of our algorithm improves to $O(n^2 \sqrt{\log n})$. More generally,  if the IDs are taken from $\{1,\ldots,n^s\}$ with $1\leq s<4$, then the round complexity of our algorithm becomes $O(n^{1+s} \sqrt{\log n})$.

\medskip

The rest of the section is entirely dedicated to the proof of \Cref{thm:nd}. Before entering into the details of the proof, let us establish a couple of statements that will be useful for the proof. 

\begin{lemma}\label{lem:virtvirt-to-virt}
    Let $G = (V,E)$ be an $n$-node graph, let $H$ be the virtual graph induced by a  uniquely-labeled  \bfs-clustering $(\ell,\delta)$ of~$G$, and let $K$ be the virtual graph induced by a  uniquely-labeled  \bfs-clustering $(\ell',\delta')$ of~$H$. There exists a distributed algorithm with constant awake complexity, and round complexity $O(n^2)$ that computes a  uniquely-labeled  \bfs-clustering $(\ell'',\delta'')$ of~$G$ such that the virtual graph induced by $(\ell'',\delta'')$ is~$K$. That is, every node $v\in V$ is given the two pairs $(\ell(v),\delta(v))$ and $(\ell'(\ell(v)),\delta'(\ell(v)))$ as input, and it must compute the pair $(\ell''(v),\delta''(v))$.
\end{lemma}

\begin{figure}
\centering
  \begin{subfigure}[b]{0.9\textwidth}
    \includegraphics[width=\textwidth]{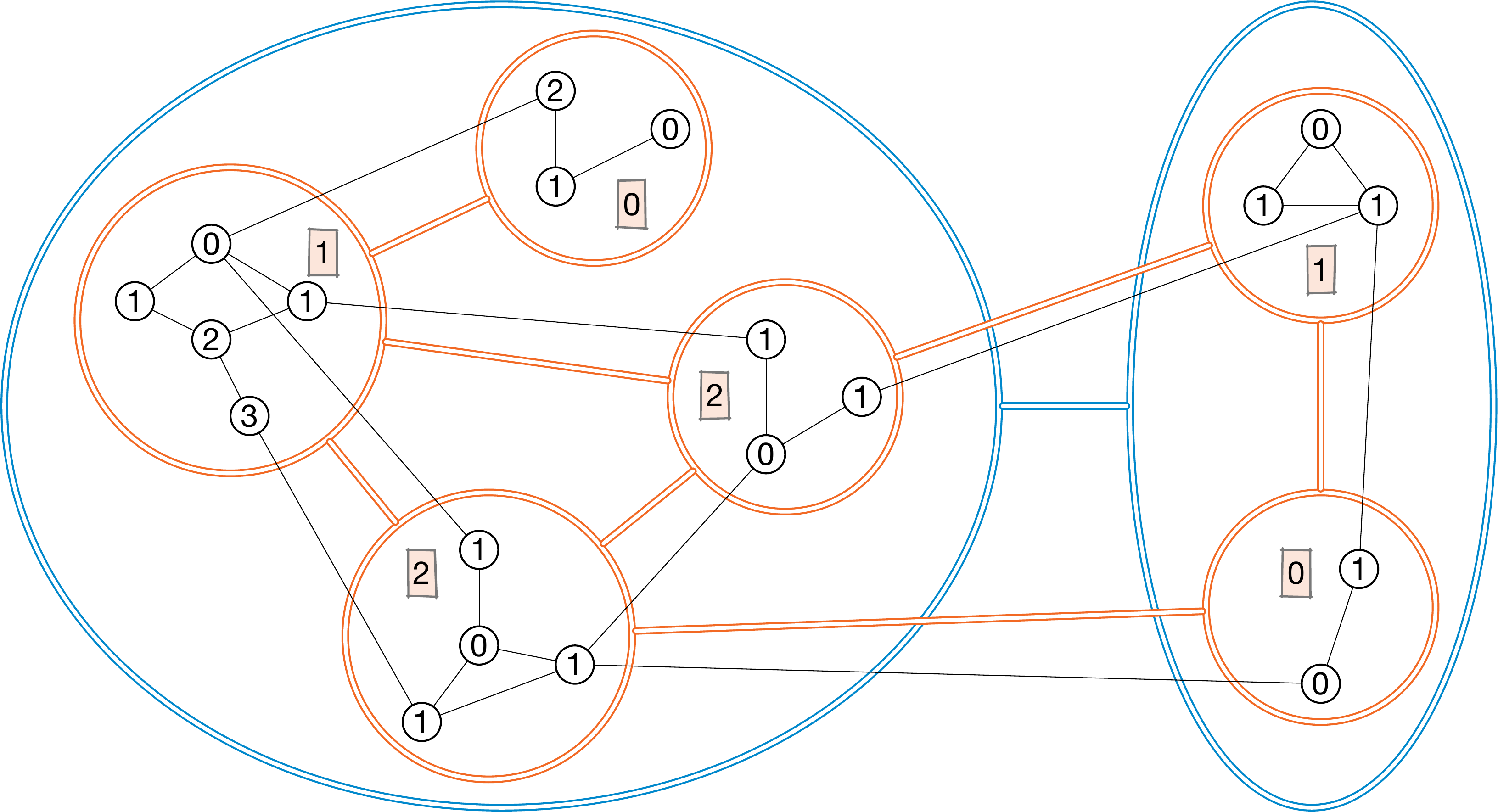}
    \caption{This figure shows an example of input for \Cref{lem:virtvirt-to-virt}. The graph $G$ is depicted in black. The virtual graph $H$ induced by a uniquely-labeled  \bfs-clustering $(\ell,\delta)$ of~$G$ is depicted in orange, where the labels contained in the black nodes denote $\delta$, and $\ell$ is represented by orange circles. The virtual graph $K$ induced by a uniquely-labeled \bfs-clustering $(\ell',\delta')$ of~$H$ is depicted in blue, where the labels contained in the orange squares nodes denote $\delta'$, and $\ell'$ is represented by blue circles.}
  \end{subfigure}
  \begin{subfigure}[b]{0.9\textwidth}
    \includegraphics[width=\textwidth]{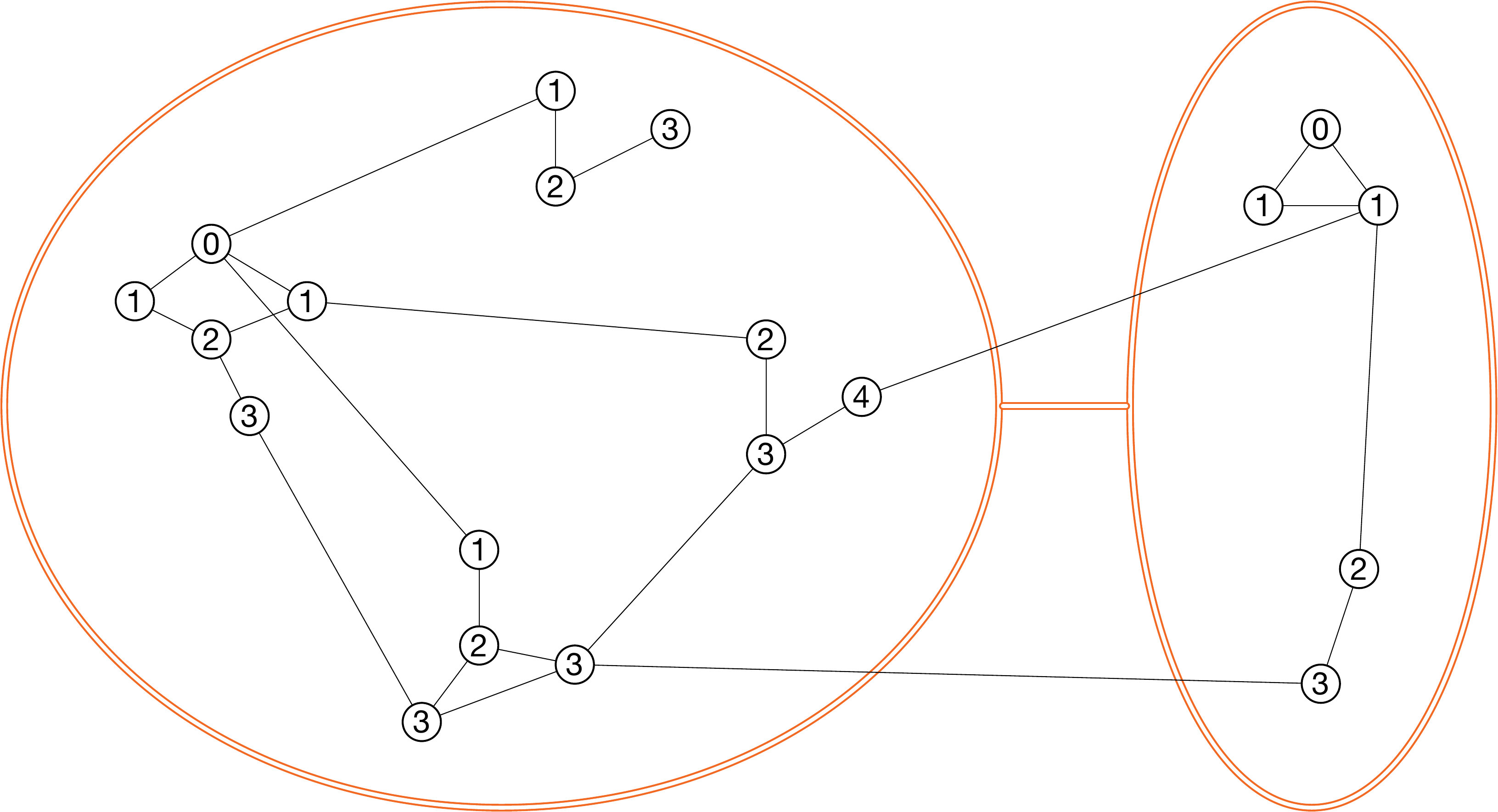}
    \caption{This figure shows the output of \Cref{lem:virtvirt-to-virt} when applied to the graph in (a). We obtain  a  uniquely-labeled  \bfs-clustering $(\ell'',\delta'')$ of~$G$ such that the virtual graph induced by $(\ell'',\delta'')$ is~$K$.}
  \end{subfigure}
  \caption{An example of application of \Cref{lem:virtvirt-to-virt}.}
  \label{fig:cluster-of-clusters}
\end{figure}

\begin{proof}
A visual representation of \Cref{lem:virtvirt-to-virt} is depicted in \Cref{fig:cluster-of-clusters}.
Every node $v\in V$ merely sets
$$\ell''(v)=\ell'(\ell(v)),$$ 
that is, its cluster is simply resulting from merging the clusters of $(\ell,\delta)$ which are given the same label in $(\ell',\delta')$. The issue is computing a distance value $\delta''(v)$ for~$v$. For this purpose, one has first to determine which nodes are roots in the  uniquely-labeled  \bfs-clustering $(\ell'',\delta'')$, i.e., which nodes $v$ satisfy $\delta''(v)=0$. We set 
$$\delta''(v)=0 \iff \delta(v)=0 \;\mbox{and}\; \delta'(\ell(v))=0.$$
That is, a node $v$ with $\ell(v)=i$ is root in $(\ell'',\delta'')$ if and only if it is root of its cluster $G_i$ in $(\ell,\delta)$, and $G_i$ is root in $(\ell',\delta')$. To compute the value $\delta''(v)$ of every node~$v$, the nodes of $G$ collectively perform the following operations whose purpose is, for every node~$v$, to learn the entire structure of the subgraph of $G$ induced by the set of nodes $\{u\in V\mid \ell''(u)=\ell''(v)\}=\{u\in V\mid \ell'(\ell(u))=\ell'(\ell(v))\}$, including the IDs of all its nodes,  which is sufficient for computing $\delta''(v)$ that is merely equal to the distance of $v$ to the root (of cluster $\ell''(v)$) in this subgraph.   

First, for each vertex $i$ of $H$, all the nodes $v\in V$ such that $\ell(v) = i$ recover the structure of the subgraph $G_i$ of $G$ induced by nodes labeled~$i$ by~$\ell$.  This is merely done by performing in parallel, within each cluster~$i$, a convergecast followed by a broadcast scheduled thanks to the function~$\delta$. More precisely, every node $v$ with $\ell(v) = i$ sets a message $M_v$ containing all edges of $E$ between $v$ and the neighbors $u$ of $v$ in $G$ satisfying $\ell(u) = i$. Moreover, among all these nodes~$u$, node $v$ elects as parent $p(v)$ any node satisfying $\delta(u)<\delta(v)$. (If no such node exists, i.e., if $v$ is the root of its cluster, then $v$ sets $p(v)=\bot$.) At this point, we can apply \Cref{lem:broad-converge-cast} for allowing the roots, i.e., all nodes $v$ such that $\delta(v)=0$, to acquire the entire structure of their clusters by a convergecast operation (this structure is merely the union of all messages $M_v$ received by each root). Then all nodes can acquire the entire structure of their clusters by a broadcast of this structure from the roots within each cluster in parallel. The awake complexity of this computation is constant, and the overall computation lasts $O(n)$ rounds.

To then allow every node $v\in V$ with label $\ell''(v)=j$, $j\in V(K)$, to recover the structure of the subgraph $H_j=\bigcup_{\ell'(i)=j}G_i$ of $G$ induced by all nodes $u\in V$ satisfying $\ell''(u)=j$, we are first setting a parent function $p'$ among the clusters $G_i$, $i\in V(H)$, of $(\ell,\delta)$. 
Every node $v$ set a message $M'_v$ containing all its incident edges $\{u,v\}$ such that $\ell''(u)=\ell''(v)$ but $\ell'(u)\neq \ell'(v)$, plus the distance $\delta'(\ell(u))$ for every such edge $\{u,v\}$. By a convergecast operation, the root of every cluster $G_i$ can select a neighboring cluster $G_{i'}$ with $\ell'(i')=\ell'(i')$ and $\delta'(i')<\delta'(i)$, and set $p'(i)=i'$. It also selects an edge between the two clusters $G_i$ and $G_{i'}$. All these information are broadcasted to all the nodes in~$G_i$.  Thanks to \Cref{lem:broad-converge-cast}, the awake complexity of all these operations is constant, and they lasts $O(n)$ rounds.

To allow each node of $H_j=\cup_{\ell'(i)=j}G_i$ to acquire the structure of~$H_j$, it suffices to simulate in $G$ the execution of a convergecast-broadcast protocol $\mathcal{P}$ in the subgraph of $H$ induced by all vertices $i$ such that $\ell'(i)=j$, by applying  \Cref{lem:broad-converge-cast} using the parent relation $p'$ and the distance $\delta'$. For this purpose, every virtual round of $\mathcal{P}$ in $H$ is replaced by a sequence of $O(n)$ rounds for allowing all the nodes of every cluster $G_i$ of $H_j$ to gather all the information exchanged with neighboring clusters during the previous virtual round. Note that, given $n$, $\delta(v)$ and $\delta'(\ell(v))$, every vertex can compute the set of rounds during which is must be awake for performing the several convergecast-broadcast operations occurring during the simulation of protocol~$\mathcal{P}$. 

The total number of rounds required for this simulation is therefore $O(n^2)$ since each of the $O(n)$ virtual rounds of $\mathcal{P}$ in $H$ requires $O(n)$ rounds to be simulated in~$G$. In $\mathcal{P}$, every vertex  $i$ of $H$ is activated $O(1)$ times, that is, each cluster $G_i$ is awake a constant number of virtual rounds. During a virtual round, each node $v\in V$ is awaken during $O(1)$ rounds. Therefore, every node $v$ is awake a constant number of rounds in total, and thus the total awake complexity of the protocol is constant. 
\end{proof}

We now state a result that will be the main ingredient for proving \Cref{thm:nd}. This lemma states that, with awake complexity $O(\log^* n)$,  one can compute  a colored \bfs-clustering $(\gamma,\delta)$ of any $n$-node graph $G = (V,E)$ with desirable properties. Namely, the clusters with small colors, i.e., with colors in a range $\{1,\dots,ab^2\}$ for well suited parameters $a$ and $b$ of our construction, are reduced to a single node. For clusters with large colors, i.e., colors larger than $ab^2$, the colored \bfs-clustering is actually a uniquely-labeled \bfs-clustering. Moreover, the number of clusters with color $>ab^2$ is small, that is, it does not exceed $n/b$ in $n$-node networks. This later property will allow us to proceed by induction in $O(\log_b n)$ phases for establishing \Cref{thm:nd}. For $b=2^{\sqrt{\log n}}$, the number of phases is $O(\log n/\log b)=O(\log n/\sqrt{\log n})= O(\sqrt{\log n})$. 

\begin{lemma}\label{lem:onestep}
There exists an integer~$a>0$ such that, for every integer~$b>0$, there exists an algorithm parametrized by~$b$, with awake complexity $O(\log^* n)$ and round complexity $O(n^4)$ which, given any $n$-node graph $G = (V,E)$, computes a colored \bfs-clustering $(\gamma,\delta)$ such that
\begin{itemize}
\item the pair $(\gamma,\delta)$ restricted to the subgraph induced by $\{v\in V\mid \gamma(v)\in \{1,\dots,a \cdot b^2\}\}$ is a colored \bfs-clustering, and, for every $v\in V$ with $\gamma(v)\in \{1,\dots,a \cdot b^2\}$, $\delta(v)=0$ (i.e., $v$ is alone in its cluster);
\item the pair $(\gamma,\delta)$ restricted to the subgraph induced by $\{v\in V\mid \gamma(v)>a \cdot b^2\}$ is a uniquely-labeled \bfs-clustering with at most $n/b$ clusters.
\end{itemize}
\end{lemma}

\noindent Let us first show how to prove \Cref{thm:nd} assuming that \Cref{lem:onestep} holds.

\begin{figure}[h]
\centering
    \includegraphics[width=0.7\textwidth]{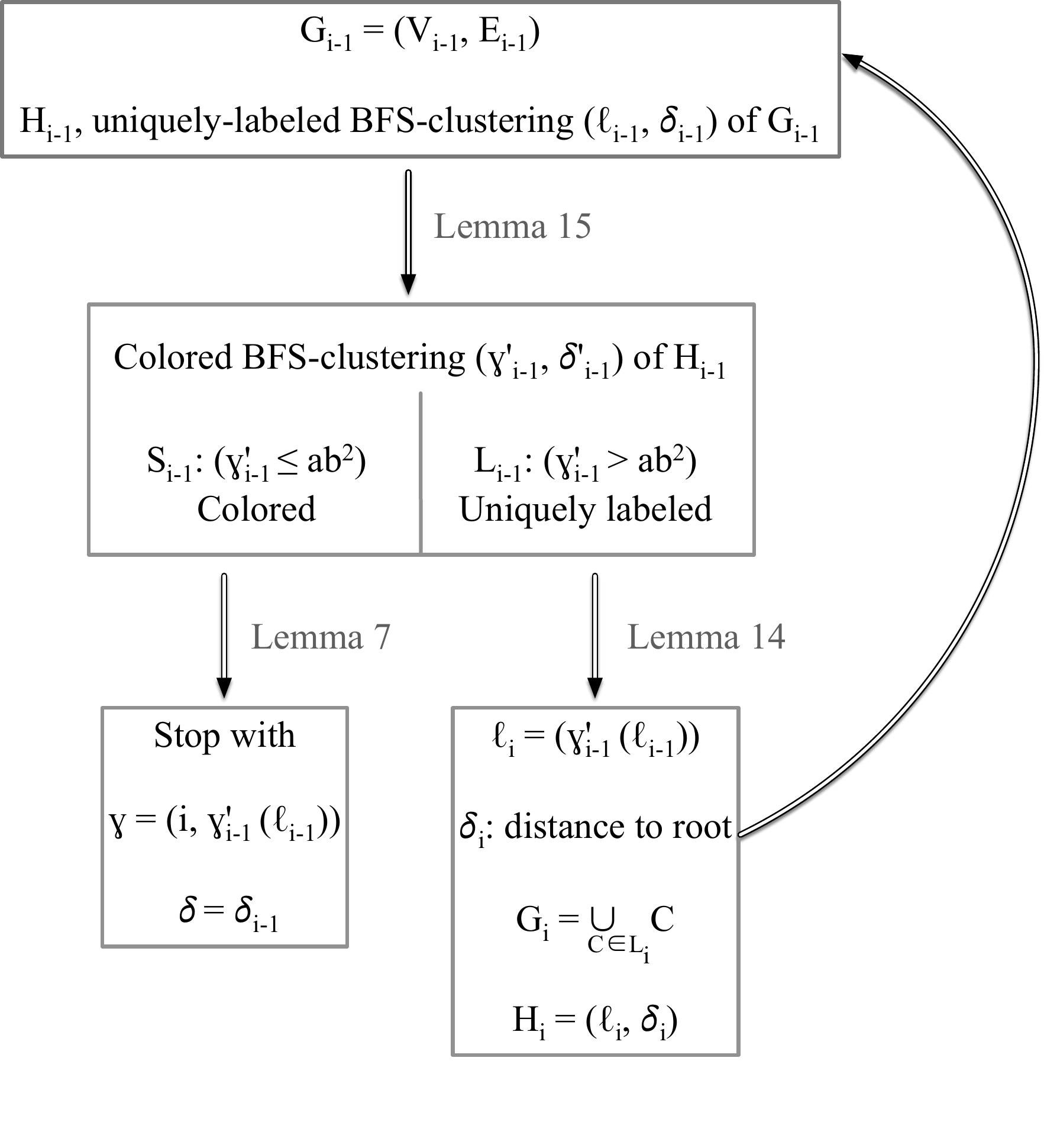}
    \caption{Sum up of the Clustering Algorithm of~\Cref{thm:nd}. We compute iteratively a clustering of the remaining clusters. Some parts get their final cluster at each iteration, while the other goes back into the loop.}
  \label{fig:proof-thm13}
\end{figure}

\begin{proof}[Proof of \Cref{thm:nd}]
Let $k=2 \sqrt{\log n}$, and $b=2^{\sqrt{\log n}}$. Let us consider the trivial  uniquely-labeled  \bfs-clustering $(\ell_0,\delta_0)$ of $G$ obtained by assigning to each node $v$ the label~$\ell_0(v)=\mbox{ID}(v)$, and $\delta_0(v)=0$. Let $H_0$ be the virtual graph associated to this clustering. The desired  colored  \bfs-clustering $(\gamma,\delta)$ is computed by a sequence of $k$ phases, labeled by $i=1,\dots,k$ (see~\Cref{fig:proof-thm13} for an illustration of the iterative process). Initially, let $G_0 = G$. For every $i\geq 1$, Phase~$i$ consists of the following. We assume given a virtual graph $H_{i-1}$ associated to some  uniquely-labeled  \bfs-clustering $(\ell_{i-1},\delta_{i-1})$ of a subgraph $G_{i-1}=(V_{i-1},E_{i-1})$ of $G$ induced by $V_{i-1}\subseteq V$. That is, every $v\in V_{i-1}$ knows the pair $(\ell_{i-1}(v),\delta_{i-1}(v))$. Phase~$i$ consists of the following.   
    
Thanks to \Cref{lem:onestep} applied for $b=2^{\sqrt{\log n}}$, a colored \bfs-clustering $(\gamma'_{i-1},\delta'_{i-1})$ of $H_{i-1}$ can be computed, with awake complexity $O(\log^* n)$, and round complexity~$O(n^4)$ on the virtual graph~$H_{i-1}$. Let 
\[
S_{i-1}=\{C\in V(H_{i-1})\mid\gamma'_{i-1}(C)\in\{1,\dots,ab^2\}\}. 
\]
and 
\[
L_{i-1}=\{C\in V(H_{i-1})\mid\gamma'_{i-1}(C)>ab^2\}, 
\]
The colored \bfs-clustering $(\gamma'_{i-1},\delta'_{i-1})$ of $H_{i-1}$ has the property that, for each vertex $C\in S_{i-1}$, the cluster of $C$ consists of $C$ only.  On the other hand, restricted to the vertices $C\in L_{i-1}$, $(\gamma'_{i-1},\delta'_{i-1})$ is a  uniquely-labeled \bfs-clustering. Our aim is to simulate  on $G_{i-1}$ the computation of $(\gamma'_{i-1},\delta'_{i-1})$ in $H_{i-1}$ so that every node $v\in V_{i-1}$ computes the pair 
\[
(\gamma'_{i-1}(\ell_{i-1}(v)),\delta'_{i-1}(\ell_{i-1}(v))).
\]
This can be achieved by a mere application of  \Cref{lem:simulation-in-virtual-graph}, with awake complexity $O(\log^* n)$, and round complexity~$O(n^5)$. 
At this point, every node $v\in V_{i-1}$ satisfying $\gamma'_{i-1}(\ell_{i-1}(v))\in \{1,\dots,ab^2\}$, sets 
\[
\gamma(v)=(i,\gamma'_{i-1}(\ell_{i-1}(v))), \;\mbox{and} \; \delta(v)=\delta_{i-1}(v),
\]
and terminates. For all the other nodes, i.e., for all the nodes of $G_{i-1}$ whose clusters in $(\gamma'_{i-1},\delta'_{i-1})$ belong to~$L_{i-1}$, one can apply \Cref{lem:virtvirt-to-virt}, stating  that the nodes of $G_{i-1}$ can compute a uniquely-labeled \bfs-clustering $(\ell_i,\delta_i)$ such that, for every $v\in V_{i-1}$, 
\[
\ell_i(v)=\gamma'_{i-1}(\ell_{i-1}(v)),
\]
and $\delta_i(v)$ is its distance to the root of its cluster. The graph $G_i$ is the subgraph of $G$ induced by $V_i=\{v\in V_{i-1}\mid \ell_i(v)  \text{ is defined
and } \ell_i(v)>ab^2\}$.
The graph $H_i$ is the virtual graph of the uniquely-labeled \bfs-clustering $(\ell_i,\delta_i)$ of~$G_i$. Note that, thanks to \Cref{lem:onestep}, $|V(H_{i})| \le |V(H_{i-1})|/b$.

The total awake complexity of the $k$ iterations is $O(\sqrt{\log n} \cdot \log^* n)$, and these $k$ iterations consume $O(n^5 \sqrt{\log n})$ rounds in total. After $k$ iterations of the above sequence of instructions, the virtual graph $H_k$ is empty. Indeed, the number of remaining nodes  after $k$ iterations  can be upper bounded by
\[
\frac{n}{b^{k}} =  \frac{ n }{ 2^{2 \sqrt{\log n}  \cdot \sqrt{\log n}}} = \frac{n}{ 2^{2 \cdot \log n}} =\frac{1}{n}<1.
\]
It follows that, after $k$ iterations, all nodes of $G$ have terminated. The number of colors used by our construction can be upper bounded by 
\[
k \cdot a \cdot b^2 = 2^{O(\sqrt{\log n})},
\]
as desired. It remains to show correctness, i.e., that $(\gamma,\delta)$ is indeed a colored \bfs-clustering of~$G$.  

By construction, nodes assigned to clusters at different iterations have different colors, and therefore two clusters created at different iterations have different colors. Let $C\subseteq V$ be a cluster created at some iteration $i\in\{1,\dots,k\}$. Every node $v\in C$ of this cluster is labeled $\gamma(v)=(i,\gamma'_{i-1}(\ell_{i-1}(v)))$. The cluster $C$ corresponds to vertex $C\in V(H_{i-1})$ with $\gamma'_{i-1}(C)\in\{1,\dots,ab^2\}$, where $H_{i-1}$ is the virtual graph corresponding to the uniquely-labeled \bfs-clustering $(\ell_{i-1},\delta_{i-1})$ of $G_{i-1}$. In the clustering $(\gamma'_{i-1},\delta'_{i-1})$, $C$~is its own cluster. It follows that the setting of $\delta(v)=\delta_{i-1}(v)$ satisfies the property of a \bfs-clustering. 

Finally, let us consider two nodes $v$ and $v'$ of $G$ belonging to clusters $C$ and $C'$ of~$\gamma$, with 
\[
\gamma(v)=(i,\gamma'_{i-1}(\ell_{i-1}(v)))=\gamma(v')=(i,\gamma'_{i-1}(\ell_{i-1}(v'))).
\]
In this case, $C$ and $C'$ were clusters of $G_{i-1}$, and, in the clustering $\gamma'_{i-1}$, $C$ and $C'$ are singletons.  If there is an edge between $v$ and $v'$ in $G$, then since $(\ell_{i-1},\delta_{i-1})$ is a uniquely-labeled \bfs-clustering of~$G_{i-1}$, we have that $C=C'$. It follows that $(\gamma,\delta)$ is a colored  \bfs-clustering, as desired. 
\end{proof}

For completing the proof of \Cref{thm:nd}, it is thus sufficient to prove \Cref{lem:onestep}. That is, we will provide an algorithm that computes a  uniquely-labeled  \bfs-clustering satisfying the requirements of \Cref{lem:onestep}, with awake complexity $O(\log^* n)$, and round complexity~$O(n^4)$.

\begin{proof}[Proof of \Cref{lem:onestep}]

\begin{figure}
\centering
\begin{subfigure}[b]{.48\textwidth}

\begin{tikzpicture}[scale=0.7]
   \tikzstyle{circlenode}=[draw,circle,minimum size=70pt,inner sep=0pt]
    \tikzstyle{whitenode}=[draw,circle,fill=white,minimum size=20pt,inner sep=0pt]
    \tikzstyle{blacknode}=[draw=black,circle=black,fill=black,minimum size=15pt,inner sep=0pt]
    \tikzstyle{nonode}=[draw=white,circle=red,fill=white,minimum size=10pt,inner sep=0pt]

\draw (0,0) node[whitenode] (a1)   {13};
\draw (2,0) node[whitenode] (a2)   {18};
\draw (-2,1.5) node[whitenode] (a3)   {145};
\draw (-2,-1.5) node[whitenode] (a4)   {165};
\draw (2,1.5) node[whitenode] (a5)   {120};
\draw (0,1.5) node[whitenode] (a6)   {136};
\draw (4,1.5) node[whitenode] (a7)   {113};
\draw (2,-1.5) node[whitenode] (a8)   {16};
\draw (0,-1.5) node[whitenode] (a10)   {121};
\draw (4,0) node[whitenode] (a14)   {105};
\draw (6,1.5) node[whitenode] (a15)   {102};
\draw (6,-1.5) node[whitenode] (a16)   {142};
\draw (6,-3) node[whitenode] (a17)   {113};
\draw (4,-1.5) node[whitenode] (a18)   {101};
\draw (8,0) node[whitenode] (a19)   {121};
\draw (8,1.5) node[whitenode] (a20)   {135};
\draw (6,0) node[whitenode] (a21)   {145};
\draw (0,-3) node[whitenode] (a22)   {126};
\draw (-2,-3) node[whitenode] (a23)   {133};
\draw (2,-3) node[whitenode] (a24)   {157};
\draw (4,-3) node[whitenode] (a25)   {128};

\draw (a1) edge node {} (a2);
\draw (a1) edge node {} (a10);
\draw (a1) edge node {} (a3);
\draw (a1) edge node {} (a4);
\draw (a2) edge node {} (a5);
\draw (a2) edge node {} (a8);
\draw (a2) edge node {} (a14);
\draw (a6) edge node {} (a5);
\draw (a7) edge node {} (a5);
\draw (a10) edge node {} (a8);
\draw (a10) edge node {} (a4);
\draw (a15) edge node {} (a14);
\draw (a16) edge node {} (a17);
\draw (a16) edge node {} (a18);
\draw (a16) edge node {} (a19);
\draw (a20) edge node {} (a19);
\draw (a20) edge node {} (a15);
\draw (a21) edge node {} (a19);
\draw (a4) edge node {} (a23);
\draw (a22) edge node {} (a23);
\draw (a24) edge node {} (a25);
\draw (a24) edge node {} (a8);
\draw (a25) edge node {} (a17);
\draw (a25) edge node {} (a8);

\draw (a1) edge [color=red,->, loop above, distance=.5cm] node {} (a1);
\draw (a18) edge [color=red,->, loop above, distance=.5cm] node {} (a18);
\draw (a3) edge [color=red,->, bend right] node {} (a1);
\draw (a4) edge [color=red,->, bend left] node {} (a1);
\draw (a10) edge [color=red,->, bend right] node {} (a1);
\draw (a8) edge [color=red,->] node {} (a1);
\draw (a2) edge [color=red,->, bend left] node {} (a1);
\draw (a5) edge [color=red,->, bend left] node {} (a2);
\draw (a14) edge [color=red,->, bend left] node {} (a2);
\draw (a6) edge [color=red,->] node {} (a2);
\draw (a7) edge [color=red,->] node {} (a2);
\draw (a15) edge [color=red,->] node {} (a2);
\draw (a19) edge [color=red,->] node {} (a18);
\draw (a17) edge [color=red,->] node {} (a8);
\draw (a16) edge [color=red,->, bend left] node {} (a18);
\draw (a21) edge [color=red,->, bend left] node {} (a19);
\draw (a20) edge [color=red,->, bend right] node {} (a15);
\draw (a23) edge [color=red,->] node {} (a1);
\draw (a22) edge [color=red,->, loop above, distance=.5cm] node {} (a22);
\draw (a24) edge [color=red,->, bend left] node {} (a8);
\draw (a25) edge [color=red,->, bend left] node {} (a8);
\end{tikzpicture}
\caption{Parent selection}
\end{subfigure}
\hfill
\begin{subfigure}[b]{.48\textwidth}
\begin{tikzpicture}[scale=0.7]
   \tikzstyle{circlenode}=[draw,circle,minimum size=70pt,inner sep=0pt]
    \tikzstyle{whitenode}=[draw,circle,fill=white,minimum size=20pt,inner sep=0pt]
    \tikzstyle{greynode}=[draw,circle,fill=black!20,minimum size=20pt,inner sep=0pt]
    \tikzstyle{blacknode}=[draw=black,circle=black,fill=black,minimum size=15pt,inner sep=0pt]
    \tikzstyle{nonode}=[draw=white,circle=red,fill=white,minimum size=10pt,inner sep=0pt]

\draw (0,0) node[whitenode] (a1)   {0};
\draw (2,0) node[whitenode] (a2)   {26};
\draw (-2,1.5) node[whitenode] (a3)   {26};
\draw (-2,-1.5) node[whitenode] (a4)   {26};
\draw (2,1.5) node[whitenode] (a5)   {36};
\draw (0,1.5) node[whitenode] (a6)   {37};
\draw (4,1.5) node[whitenode] (a7)   {37};
\draw (2,-1.5) node[whitenode] (a8)   {27};
\draw (0,-1.5) node[whitenode] (a10)   {26};
\draw (4,0) node[whitenode] (a14)   {36};
\draw (6,1.5) node[whitenode] (a15)   {37};
\draw (6,-1.5) node[greynode] (a16)   {202};
\draw (6,-3) node[whitenode] (a17)   {33};
\draw (4,-1.5) node[greynode] (a18)   {0};
\draw (8,0) node[greynode] (a19)   {203};
\draw (8,1.5) node[whitenode] (a20)   {204};
\draw (6,0) node[greynode] (a21)   {242};
\draw (0,-3) node[greynode] (a22)   {0};
\draw (-2,-3) node[whitenode] (a23)   {27};
\draw (2,-3) node[whitenode] (a24)   {32};
\draw (4,-3) node[whitenode] (a25)   {32};

\draw (a1) edge node {} (a2);
\draw (a1) edge node {} (a10);
\draw (a1) edge node {} (a3);
\draw (a1) edge node {} (a4);
\draw (a2) edge node {} (a5);
\draw (a2) edge node {} (a8);
\draw (a2) edge node {} (a14);
\draw (a6) edge node {} (a5);
\draw (a7) edge node {} (a5);
\draw (a10) edge node {} (a8);
\draw (a10) edge node {} (a4);
\draw (a15) edge node {} (a14);
\draw (a16) edge node {} (a17);
\draw (a16) edge node {} (a18);
\draw (a16) edge node {} (a19);
\draw (a20) edge node {} (a19);
\draw (a20) edge node {} (a15);
\draw (a21) edge node {} (a19);
\draw (a4) edge node {} (a23);
\draw (a22) edge node {} (a23);
\draw (a24) edge node {} (a25);
\draw (a24) edge node {} (a8);
\draw (a25) edge node {} (a17);
\draw (a25) edge node {} (a8);

\draw (a1) edge [color=red,->, loop above, distance=.5cm] node {} (a1);
\draw (a18) edge [color=red,->, loop above, dashed, distance=.5cm] node {} (a18);
\draw (a3) edge [color=red,->, bend right] node {} (a1);
\draw (a4) edge [color=red,->, bend left] node {} (a1);
\draw (a10) edge [color=red,->, bend right] node {} (a1);
\draw (a2) edge [color=red,->, bend left] node {} (a1);
\draw (a5) edge [color=red,->, bend left] node {} (a2);
\draw (a14) edge [color=red,->, bend left] node {} (a2);
\draw (a6) edge [color=blue,->, bend left] node {} (a5);
\draw (a17) edge [color=blue,->, bend left] node {} (a25);
\draw (a16) edge [color=red,->, bend left, dashed] node {} (a18);
\draw (a21) edge [color=red,->, bend left, dashed] node {} (a19);
\draw (a20) edge [color=red,->, bend right] node {} (a15);
\draw (a23) edge [color=blue,->, bend left] node {} (a4);
\draw (a22) edge [color=red,->, loop above, dashed, distance=.5cm] node {} (a22);
\draw (a24) edge [color=red,->, bend left] node {} (a8);
\draw (a25) edge [color=red,->, bend left] node {} (a8);

\draw (a8) edge [color=blue,->, bend right] node {} (a2);
\draw (a7) edge [color=blue,->, bend right] node {} (a5);
\draw (a15) edge [color=blue,->, bend right] node {} (a14);
\draw (a19) edge [color=blue,->, bend left, dashed] node {} (a16);
\end{tikzpicture}
\caption{Cluster decomposition}
\end{subfigure}
\caption{In (a), we have a distance-2 coloring $c_1$ of the nodes, with $b=3$ and $k=100$. Notice that nodes of degree $\le3$ have added 100 to their colors. In red, each node $u$ is connected to $p_1(u)$ (with a self loop if $p_1(u)=\bot$).
In (b), each node $u$ has computed its new color $c_2(u)$ (that is not a proper coloring, it is a decreasing coloring from node to parent) and its new parent $p_2(u)$ (in blue if it differs from $p_1(u)$). The trees formed with a root of degree at least 4 form the new clusters. (Grey) Nodes in a tree of root of degree at most 3 become singleton clusters (dotted tree edges to illustrate that the parent relation is forgotten). These nodes will compute some $ab^2$ coloring.}
\label{fig:pointers}
\end{figure}
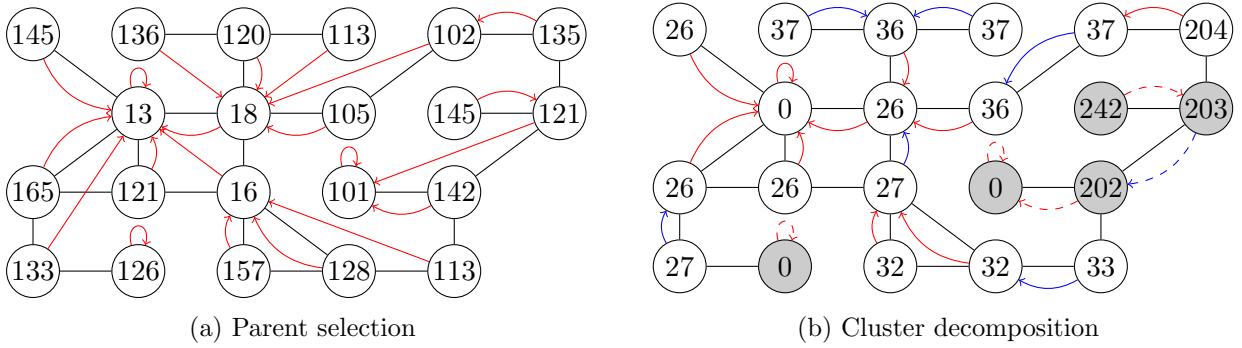

    The algorithm works as follows (see \Cref{fig:pointers}). First, the nodes of the $n$-node graph $G$ collectively compute a proper coloring of $G^2$ with a palette of $k=O(n^4)$ colors, where $G^2$ is the square of~$G$, that is, the graph with the same set of nodes as~$G$, and where two nodes are connected by an edge if they are at distance at most~$2$ in~$G$. This can be done in $O(\log^* n)$ rounds in the \local model by using, e.g.,  Linial's coloring algorithm~\cite{linial92}. Let us denote by $c_0:V\to\{1,\dots,k\}$ this coloring. 
    
    Note that if the IDs assigned to the nodes are
    taken from the range $\{1,\ldots,n^s\}$ for some $s\geq 1$, then this step could be replaced by directly taking the IDs of the nodes as the desired distance-$2$ $k$-coloring, with $k=n^s$, in zero rounds.

    Let $b>0$ be an integer. Every node $v$ with $\deg_G(v)\leq b$ increases its color $c_0(v)$ by adding $k$ to it. That is, the new color $c_1(v)$ of $v$ is set as $c_1(v)=c_0(v)+k$. Every node $v$ with $\deg_G(v)>b$ sets $c_1(v)=c_0(v)$.  Hence, for every $v\in V$, we have $1\leq c_1(v)\leq 2k$. Note that the  colors of the nodes of degree at most $b$ are in the range $\{k+1,\dots,2k\}$, whereas the colors of the nodes of degree larger than~$b$ are in the range $\{1,\dots,k\}$. 

    Let us now set a ``parent pointer'' $p_1(v)$ at every node $v\in V$. Let us denote by $N(v)$ the (open) neighborhood of~$v$, and by $N^2(v)$ the set of nodes at distance exactly 2 from $v$ in~$G$. Three cases are considered, depending on the colors of the nodes in $N(v)\cup N^2(v)$. 

    \begin{itemize}
        
\item If every node $u\in N(v)\cup N^2(v)$ satisfies $c_1(u) > c_1(v)$, then $p_1(v) = \bot$. A ``shift'' variable $b(v)$ is also set to $b(v) = \bot$ in this case.

\item If there exists $u\in N(v)$ satisfying $c_1(u) < c_1(v)$, then $p_1(v) = u_{min}$ where $u_{min}$ is the node with smallest color in $N(v)$, breaking ties arbitrarily.  In this case, the shift at $v$ is set to $b(v) = 0$. 
    
\item If none of the above two conditions are fulfilled, that is, if every $u\in N(v)$ satisfies $c_1(u) > c_1(v)$, but there exists $u\in N^2(v)$ such that $c_1(u)<c_1(v)$, then $p_1(v)$ is set as the node with smallest color in $N^2(v)$, breaking ties arbitrarily. In this case, the shift at $v$ is set to $b(v) = 1$.
 \end{itemize}

\noindent 
By construction, the collection of pointers $p_1$ induces a rooted spanning forest $F_1$ of~$G^2$ (i.e., every tree in the forest is rooted). Namely, the roots are all nodes $v$ with $p_1(v)=\bot$. Moreover, for every $v\in V$,
\[
p_1(v) \neq \bot \Longrightarrow c_1(v) > c_1(p(v)).
\]
The pointers $p_1$, and the coloring $c_1$ are now modified at every node to obtain a new pointer~$p_2$, and a new coloring~$c_2$, as follows:

\begin{itemize}
\item Every root~$v$ of a tree in~$F_1$, i.e., every node $v\in V$ with $p_1(v)=\bot$, sets its color $c_2(v) = 0$. Every other node $v$ takes  color 
$$c_2(v) = 2 \cdot c_1(p_1(v)) + b(v).$$ 
That is, every $v\in V$ with $p_1(v)\neq\bot$ copies the color of its parent, doubles it, and shifts it by~$1$ if its parent is at distance~$2$ from it in~$G$.

\item  The new pointer $p_2(v)$ at node $v\in V$ is defined as follows:
\begin{itemize}
    \item If $p_1(v) = \bot$, then $p_2(v) = \bot$;
    \item If $b(v) = 0$, then $p_2(v) = p_1(v)$;
    \item If $b(v)= 1$, then $p_2(v) = u$ where $u$ is an arbitrary node in $N(v)\cap N(p_1(v))$. 
\end{itemize} 
Note that we have $p_2(v) = \bot$ if and only if $p_1(v) = \bot$.
\end{itemize} 

\begin{claim}\label{claim:spanning-forest}
For every node $v\in V$ with $p_2(v)\neq\bot$, $c_2(v) > c_2(p_2(v))$, and the collection of pointers $p_2$ induces a rooted spanning  forest $F_2$ in~$G$.
\end{claim}

This claim will be proved later, and we carry on with the description of the algorithm by assuming that it holds.

In each rooted tree in $F_2$, the nodes collectively perform a convergecast, followed by a broadcast, in order to acquire the whole structure of the tree, including the IDs and colors $c_2$ of all its nodes. In particular, every node $v$ can compute its distance $\delta_{aux}(v)$ from the root of its tree, and the ID of this root, that we denote by~$\ell(v)$. By construction, the pair $(\ell,\delta_{aux})$ is a  uniquely-labeled  \bfs-clustering of~$G$. Note that, thanks to  \Cref{lem:broad-converge-cast}, the convergecast-broadcast operation can be done in $O(n^4)$ rounds, with constant awake complexity. 
The colored \bfs-clustering $(\gamma,\delta)$ of \Cref{lem:onestep} is obtained from $(\ell,\delta_{aux})$ as follows. 

Let $C$ be a cluster in  $(\ell,\delta_{aux})$, with a root $r$ satisfying $\deg(r)\leq b$. It must be the case that every node $v\in C\smallsetminus \{r\}$ satisfies $c_2(v)>c_2(r)$. Moreover, since the colors of the nodes with degree~$\leq b$ are larger than the colors of the nodes with degree~$>b$, we get that, for every $v\in C$, $\deg(v)\leq b$. Let $U \subseteq V$ be the subset of nodes in $G$ belonging to a cluster $C$ whose root~$r$ satisfies $\deg(r)\leq b$, and let $G[U]$ denote the subgraph of $G$ induced by the nodes in~$U$. By construction, the graph $G[U]$ has maximum degree at most~$b$. 
The nodes of $G[U]$ collectively compute a proper coloring $\gamma$ of $G[U]$ with $O(b^2)$ colors using Linial's coloring algorithm~\cite{linial92}. 
Let us fix $a$ as the smallest integer such that the number of colors produced by Linial's  algorithm is at most $a\cdot b^2$. The round complexity (and therefore the awake complexity) of this algorithm is  $O(\log^* n)$. At this point, every $v \in U$ becomes part of a cluster $\gamma(v)$ consisting of $v$ alone, and $v$ sets $\delta(v)=0$. 

Every node $v\in V\smallsetminus U$, i.e., every node $v$ belonging to a cluster in $(\ell,\delta_{aux})$ with a root of degree at least $b+1$ sets $\gamma(v)=\ell(v)+a\cdot b^2$, and $\delta(v)=\delta_{aux}(v)$. 

This completes the construction of the desired clustering $(\gamma,\delta)$. The total awake complexity of all the above operations is  $O(\log^* n)$, and its round complexity is~$O(n^4)$. 
It remains to show correctness (including the proof of Claim~\ref{claim:spanning-forest}). 

By construction, the colored \bfs-clustering $(\gamma,\delta)$ satisfies the statement of \Cref{lem:onestep} regarding the nodes in $U$ and the nodes in $V\smallsetminus U$.  
What is left to prove is that, if we restrict the \bfs-clustering $(\gamma,\delta)$ to the subgraph induced by $\{v\in V\mid \gamma(v)>a \cdot b^2\}$, we obtain at most $n/b$ clusters. For this purpose, we upper bound the number of roots.
For each cluster $C$ in $(\ell,\delta_{aux})$, the root $r$ of $C$ satisfies $\deg(r)>b$. Moreover, for each neighbor $v$ of~$r$, $v$~cannot be neighbor of a root $r'$ of another cluster $C'$ in $(\ell,\delta_{aux})$ with $\deg(r')>b$. This is because the roots in $(\ell,\delta_{aux})$ are local minima in~$G^2$, and thus they are at mutual distance at least $3$ in~$G$.
It follows that, for each root $r$ in $(\ell,\delta_{aux})$ with $\deg(r)>b$, we can charge all $r$'s neighbors to~$r$. These neighbors are at least $b+1$, and they are not roots. Note that each node $v\in V$ is charged to at most one root. 
We thus get that the number of roots $r$ with $\deg(r)>b$ are at most $n / b$, as desired. This completes the proof of \Cref{lem:onestep}, assuming Claim~\ref{claim:spanning-forest} holds. 
\end{proof}

It just remains to prove Claim~\ref{claim:spanning-forest}.

\begin{proof}[Proof of Claim~\ref{claim:spanning-forest}]
We first show that, for every node $v\in V$ with $p_2(v)\neq \bot$, $c_2(v) > c_2(p_2(v))$. For each such node $v$, we have $b(v)\in\{0,1\}$. We analyse the two cases separately.
\begin{itemize}
    \item Let $v\in V$ with $b(v) = 0$, i.e., $p_2(v) = p_1(v)$. In this case, node $p_1(v)$ may either be a root or not. If $p_1(p_1(v)) = \bot$, then 
    \[
    c_2(p_2(v)) = c_2(p_1(v)) = 0, \;\mbox{and}\; c_2(v) = 2\cdot c_1(p_1(v)) > 0,
    \]
    from which it follows that $c_2(p_2(v)) < c_2(v)$.
    If $p_1(p_1(v)) \neq \bot$, say $p_1(p_1(v)) = u$, then $c_1(p_1(v)) > c_1(u)$. Moreover,  
    \[
    c_2(p_2(v)) = c_2(p_1(v)) = 2\cdot c_1(u), \; \mbox{and}\; c_2(v) = 2 \cdot c_1(p_1(v)),
    \]
    from which it again follows that $c_2(p_2(v)) < c_2(v)$.

    \item Let $v\in V$ with $b(v) = 1$. Let $p_2(v) = u\in N(v)\cap N(p_1(v))$. By the definition of $b(v)$, $c_1(u) > c_1(v)$, and $b(u) = 0$. Since $u$ is a neighbor of~$p_1(v)$,  $p_1(v) \in N(u)$,  and since $p_1(v) \in N^2(v)$, we get that $p_1(v)$ has the smallest color among all nodes in~$N(u)$. Moreover, since $c_1$ is a distance-2 coloring, there are no nodes $z \in N(u)\smallsetminus \{p_1(v)\}$ that have the same color as $p_1(v)$.  By combining these facts with the fact that $b(u) = 0$, we get that $p_1(u) = p_1(v)$. As a consequence, $$c_2(v) = 2 \cdot c_1(p_1(v)) + 1 > 2 \cdot c_1(p_1(v))= 2 \cdot c_1(p_1(u)) = c_2(u) = c_2(p_2(v)).$$
\end{itemize}
So the first statement of Claim~\ref{claim:spanning-forest} holds in the two cases $b(v) = 0$ or~$1$. 

Finally, to show that the collection of pointers $p_2$ induces a spanning rooted forest $F_2$ in~$G$, it is sufficient to note that each node $v\in V$ with $p_2(v) \neq  \bot$ satisfies $c_2(v) > c_2(p_2(v))$. It follows that each such node $v$ is part of a rooted tree, and the union of all these trees spans~$G$. This completes the proof of Claim~\ref{claim:spanning-forest}, and thus the proof of \Cref{lem:onestep}.
\end{proof}

\section{Conclusion and Open Problems}

We have proved that the awake complexity of computing a colored \bfs-clustering with $2^{O(\sqrt{\log n})}$ colors is $O(\sqrt{\log n} \log^* n)$ rounds. Moreover, we can exploit such a clustering to solve any problem in the \olocal class with $O(\sqrt{\log n})$ awake complexity. We do not know whether the $\log^* n$ term in the awake complexity of computing this clustering is necessary.

\begin{oq}
Is it possible to compute a colored \bfs-clustering with $2^{O(\sqrt{\log n})}$ colors with awake complexity $O(\sqrt{\log n})$?
\end{oq}

More importantly, it would be interesting to figure out whether the $\sqrt{\log n}$ bound can be reduced for all problems in the \olocal class.

\begin{oq}
    Is it possible to solve all problems in \olocal  with $o(\sqrt{\log n})$ deterministic awake complexity?
\end{oq}

Additionally, while we focused mostly on the awake complexity of our algorithms, their round complexities are polynomial. It would thus be interesting to figure out whether similar results as the one in the paper regarding the awake complexity can ba achieved while improving the round complexity to, e.g., polylogarithmic. Similarly, it would be interesting to decrease the \emph{average} awake complexity of the nodes.

\begin{oq}
    Is it possible to solve all problems in \olocal  with $o(\log n)$ deterministic awake complexity and polylogarithmic round complexity? Is it possible to do so with  $o(\sqrt{\log n})$, or even constant \emph{average} awake complexity?
\end{oq}

Moreover, understanding how much randomization can help is also an interesting open question.

\begin{oq}
    Can we exploit randomization for solving problems in \olocal with $o(\sqrt{\log n})$ awake complexity?
\end{oq}

Finally, understanding whether our techniques can be applied to solve \emph{edge} problems such as maximal matching and $(2\Delta-1)$-edge coloring remains an interesting open question.
\begin{oq}
    Can we extend \olocal to include edge problems and solve them with deterministic sublogarithmic awake complexity?
\end{oq}

\section*{Acknowledgements}
We thank William K.\ Moses Jr.\ for pointing out that maximal matching does not belong to the class \olocal.

\urlstyle{same}
\bibliographystyle{alpha}
\bibliography{ref}

\end{document}